\documentclass{article}
\usepackage[T1]{fontenc}
\usepackage[english]{babel}
\usepackage{eso-pic}
\usepackage{graphicx, xcolor}
\usepackage{url}
\usepackage{tabularx}
\usepackage{multicol}
\usepackage{mdframed,caption,xspace}
\definecolor{codeBackground}{rgb}{0.9, 0.9, 0.8}
\usepackage{stmaryrd}
\usepackage{hyperref}
\hypersetup{
    colorlinks=true,
    linkcolor=blue,
    filecolor=magenta,      
    urlcolor=cyan,
    citecolor=blue,
    pdftitle={Proof of Useful Work},
    pdfpagemode=FullScreen,
    }
\usepackage{tikz}
\usetikzlibrary{arrows.meta, positioning}

\mdfdefinestyle{MyFrame}{%
    linecolor=black,
    outerlinewidth=2pt,
    roundcorner=20pt,
    innertopmargin=4pt,
    innerbottommargin=4pt,
    innerrightmargin=4pt,
    innerleftmargin=4pt,
        leftmargin = 4pt,
        rightmargin = 4pt,
    backgroundcolor=codeBackground
        }

\newcounter{greybox}[section]
\renewcommand{\thegreybox}{\arabic{section}.\arabic{greybox}}
\newenvironment{myalg}[1][]
{\refstepcounter{greybox}
 \begin{mdframed}[ 
    backgroundcolor=codeBackground, 
    linecolor=black, 
    linewidth=0.5pt, 
    roundcorner=5pt, 
    skipabove=10pt, 
    skipbelow=10pt, 
    innerleftmargin=10pt, 
    innerrightmargin=10pt, 
    innertopmargin=10pt, 
    innerbottommargin=10pt 
  ]
  \ifx\relax#1\relax
  \else{\centering\textbf{#1 }\par\medskip}\fi}
{\centering{\textbf{Algorithm \thegreybox}}\end{mdframed}
}

\usepackage{tikz}
\usetikzlibrary{decorations.pathmorphing}
\usepackage{titlesec}

\usepackage[margin=1in]{geometry} 
\usepackage{amsfonts,amsthm,amsmath,amssymb}
\usepackage{graphicx} 
\usepackage{mdframed}
\usepackage{lipsum}
\usepackage{enumitem}
\usepackage{tcolorbox}
\usepackage{accents}
\tcbuselibrary{theorems}
\usepackage[linesnumbered,ruled,vlined]{algorithm2e}
\usepackage{enumitem}
\usepackage{cleveref} 
\newtheorem{theorem}{Theorem}[section]

\newtheorem{lemma}[theorem]{Lemma}
\newtheorem{fact}[theorem]{Fact}

\newtheorem{definition}[theorem]{Definition}
\newtheorem{assumption}[theorem]{Assumption}

\newtheorem{claim}[theorem]{Claim}

\newtheorem{conjecture}[theorem]{Conjecture}

\newtheorem{problem}[theorem]{Problem}
\newtheorem{remark}[theorem]{Remark}
\numberwithin{equation}{section}

\newcommand{\R}{\mathbb{R}}

\newcommand{\cupow}{\texttt{cuPOW}}

 \newcommand{\cP}{\mathcal{P}}

\newcommand{\pow}{$\mathsf{PoW}$\xspace}
\newcommand{\pos}{$\mathsf{PoS}$\xspace}

\newcommand{\upow}{$\mathsf{PoUW}$\xspace}

\newtcbtheorem[number within=section]{mytheo}{Problem Definition}%
{colback=green!5,colframe=green!35!black,fonttitle=\bfseries}{th}

\newcommand{\omri}[1]{{\color{purple}#1}}

\title{
Proofs of Useful Work from \\Arbitrary Matrix Multiplication}

\author{Ilan Komargodski  \and Omri Weinstein}
\date{}

\begin{document}

\maketitle

\begin{abstract}
We revisit the longstanding open problem of implementing Nakamoto's proof-of-work (PoW) consensus based on a \emph{real-world} computational task $T(x)$ (as opposed to artificial random hashing), in a truly  permissionless setting where the \emph{miner itself}  chooses the input $x$. The challenge in designing such a Proof-of-Useful-Work (PoUW) protocol, is using the \emph{native} computation of $T(x)$ to produce a PoW certificate with prescribed hardness and with negligible computational overhead over the worst-case complexity of $T(\cdot)$ -- This  ensures malicious miners cannot ``game the system" by fooling the verifier to accept with higher probability compared to honest miners (while using similar computational resources). Indeed, obtaining a PoUW with $O(1)$-factor overhead is trivial for \emph{any} task $T$, but also useless. 

Our main result is a PoUW for the task of \emph{Matrix Multiplication} $\mathsf{MatMul}(A,B)$ of \emph{arbitrary} matrices with  $1+o(1)$ multiplicative overhead compared to na\"ive $\mathsf{MatMul}$. We conjecture that our protocol has optimal security in the sense that a malicious prover cannot obtain any significant advantage over an honest prover. This conjecture is based on reducing hardness of our protocol to the task of solving a batch of low-rank random linear equations which is of independent interest.

Since $\mathsf{MatMul}$s are the bottleneck of AI compute as well as countless industry-scale applications, this primitive suggests a concrete design of a new L1 base-layer protocol, which nearly eliminates the energy-waste of Bitcoin mining -- allowing GPU consumers to \emph{reduce their AI training and inference costs}  by ``re-using" it for blockchain consensus, in exchange for block rewards (2-for-1). This blockchain is currently under construction. 
\end{abstract}

\newpage


\newcommand{\tsk}{\mathsf{x}}
\newcommand{\sol}{\mathsf{y}}
\newcommand{\difficulty}{\kappa}
\newcommand{\GPUCryptoHash}{{H_1}}
\newcommand{\GPUSimpleHash}{H_2}
\newcommand{\rk}{\mathsf{rk}}
\newcommand{\CPUCryptoHash}{\mathsf{CPUCryptoHash}}
\newcommand{\seed}{\sigma}
\newcommand{\bcstate}{\seed}
\newcommand{\localstate}{\mathsf{nonce}} 
\newcommand{\iv}{\mathsf{IV}} 
\newcommand{\com}{\mathsf{com}} 
\newcommand{\diag}{\mathsf{diag}} 
\newcommand{\GenNoiseMat}{\mathsf{GenNoiseMat}}
\newcommand{\tileSize}{\mathsf{tile}}
\newcommand{\cd}{n}
\newcommand{\MiningMatMult}{\mathsf{FunctionalMatMult}}
\newcommand{\Tr}{\mathsf{Tr}}
\newcommand{\Noise}{\textsc{Encode}}
\newcommand{\DeNoise}{\textsc{Decode}}
\newcommand{\MatMul}{\textsc{MatMul}\xspace}
\newcommand{\E}{\mathop{\mathbb{E}}}
\newcommand{\Work}{\mathsf{Work}}
\

\section{Introduction}
The concept of proofs of work (PoW), i.e., easy-to-check proofs of computational effort, was proposed in 1992 in a seminal work of Dwork and Naor~\cite{DBLP:conf/crypto/DworkN92}. Their original motivation was to discourage spam email messages by requiring the sender to compute a moderately hard, but not intractable, function in order for the email to go through.  This idea is far more general than combating spam emails; indeed, it can be used as a method for limiting access to any resource. 

A concretely efficient and very simple instantiation of a PoW (in the random oracle model) was proposed by Back in 1997 under the name Hashcash and formalized in a paper five years later~\cite{back2002hashcash}. Since then, the concept has been used as a denial-of-service counter measure technique in a number of systems. One way to describe the Hashcash algorithm's idea is via a protocol between a prover and a verifier, wherein the verifier gets eventually convinced that the prover did a moderate amount of work. The protocol works with a cryptographic hash function $H$, such as SHA256 or SHA3. At a very high level, the verifier sends a random string $\seed$ as a challenge and the goal of the prover is to find a string $x$  such that $H(\seed\| x)$ starts with~$t$ leading 0s.\footnote{The symbol $\| $ represents string concatenation.} Verification is very easy --- just count leading 0s --- and can even be done by a human eye. The soundness of the protocol relies on the fact that it is impossible to find an $x$ as above faster than doing a brute-force search, assuming $H$ is a random oracle.

Perhaps the most well-known application of Back's PoW idea is in the Bitcoin cryptocurrency network~\cite{nakamoto2008bitcoin,GarayKL24}. Instead of using PoW to deter malicious email senders, Bitcoin's PoW is used to enable competitive mining and thereby to secure the network. Specifically, a Bitcoin miner runs a program that collects unconfirmed transactions from  the network to form a block. A block is only accepted by the network if its hash meets a certain difficulty target, (essentially) described by the number of trailing 0s needed in the hash output.

\paragraph{PoW Criticism.} Undoubtedly, the realization of the Bitcoin network and its underlying technology is a revolutionary  achievement in distributed computing, with far-reaching applications in both theory and practice. Notwithstanding, a significant setback and common point of criticism on the technology is its energy-intensive nature and exorbitant environmental costs (consuming over 2\% annual  electricity in the US alone).\footnote{\url{https://ccaf.io/cbnsi/cbeci/ghg/comparisons} (accessed 02.04.2025).} 
To make matters worse, the need for continual acquisition of  specialized hardware (ASICs) to retain  competitive mining, proliferates electronic waste generation, as these devices quickly become obsolete. 

Thus, an outstanding question in this context is whether it is necessary for the mining process to be so wasteful and have no other use beyond securing and driving the blockchain. A similar question was raised already in the 1992 work of Dwork and Naor, in the context of using PoW as a spam filter~\cite{DBLP:conf/crypto/DworkN92}: 
\begin{center}
{\fontfamily{qcr}\selectfont
``Finally, the evaluation of the pricing function serves no useful purpose, except
serving as a deterrent. It would be exciting to come up with a scheme in which
evaluating the pricing function serves some additional purpose.''
}
\end{center}

\paragraph{Proofs of Useful Work.} The above challenge is nowadays commonly referred to as a \emph{proof of useful work} (PoUW). Intuitively, a proof of work is \emph{useful} if it simultaneously satisfies the following properties:
\begin{itemize}
    \item \underline{Economic value}: 
The results of computations performed
by miners should be of interest, and moreover, someone must be
willing to pay for the result of the computation (formally, a task is $\kappa$-useful, if someone is willing to pay $\kappa\$$ for its completion).
\item \underline{Efficiency}: The algorithm used to solve the
problems is essentially optimal. This property is necessary to make sure that there is no ``waste'' in the mining process as
compared to an external setting.
\item \underline{Security}: It should be infeasible to fool the verifier into believing a problem was solved 
using bounded resources while it was not.  
\end{itemize}  

The first formal treatment of PoUW we are aware of is due to Ball et al.~\cite{BallRSV17a} wherein the authors suggest a proof of work mechanism (in the plain model) based on the conjectured hardness of various fine-grained algorithmic problems. They termed the proof of work useful because the miner could choose the instance of the task it wishes to solve by itself. Unfortunately, the overhead of the honest prover was significant, requiring the prover to solve poly-log many different instances of the task just to convince the verifier that a single instance was solved. In other words, their scheme did not satisfy the efficiency property from above and in a follow-up version of their manuscript~\cite{DBLP:conf/crypto/BallRSV18} they retracted all claims about ``usefulness.'' 

The absence of a PoUW (prior to this work)  is not due to lack of trying. 
Not only was this problem mentioned by Dwork and Naor~\cite{DBLP:conf/crypto/DworkN92}, PoUW has been the holy-grail of distributed consensus since the early days of Bitcoin. Indeed, it avoids the \emph{tradeoff} between resource-efficiency and security, present in proof-of-stake (PoS) or other energy-efficient alternatives. 
As Vitalik Buterin, co-founder of Ethereum, articulates in his 2019 blog post:\footnote{\url{https://vitalik.eth.limo/general/2019/11/22/progress.html}, Bullet 7.}  

\begin{center}
{\fontfamily{qcr}\selectfont
\begin{tabularx}{1.0\textwidth}{X}
``If either an efficiently verifiable proof-of-computation for Folding@home can be produced, or if we can find some other useful computation which is easy to verify, then cryptocurrency mining could actually become a huge boon to society, not only removing the objection that Bitcoin wastes "energy", but even being socially beneficial by providing a public good."
\end{tabularx}
}
\end{center}

In the same post, Buterin adds that the challenge of obtaining a proof of useful work is ``{\fontfamily{qcr}\selectfont probably not feasible}''. The transition of Ethereum from PoW to PoS is, to a large extent, attributed to this impossibility conjecture.  

Even a theoretical construction of PoUW was elusive. Specifically, we are not aware of any suggestion for a truly useful proof of work mechanism where the miner (and not the network) gets to choose their own instance. We refer to a manuscript of Dotan and Tochner~\cite{DBLP:journals/corr/DotanTochner} surveying prior approaches,  identifying for each one which properties of a ``non-wasteful'' and ``truly useful'' proof of work they do not satisfy.

\subsection{Our Contributions}

Our main contribution is the first construction of a \emph{truly useful} PoW protocol, for the omnipresent 
task of \emph{matrix multiplication} (MatMul). At a high level, our (GPU-compatible) protocol can re-use native $\MatMul(A,B)$ computations of \emph{arbitrary} matrices (submitted by the miner at her own discretion/application), to produce a verifiable PoW certificate obeying Nakamoto's Poisson distribution, \emph{with $1+o(1)$ multiplicative overhead} over the 
worst-case time for computing $A\cdot B$
(e.g., $(1+ o(1))n^3$ in the case of square-matrices, assuming na\"ive MatMul is the baseline).

Before we describe the protocol (called $\cupow$; see Section~\ref{sec:overview}), 
let us first review the axioms of PoUW, and argue why the $\MatMul$ primitive satisfies them. 
\begin{itemize}
    \item Economic value: Matrix multiplication is the backbone of many industry-scale applications. In particular, they are (by far) the bottleneck of AI training and inference, where giant MatMuls 
    (as large as 40K x 40K in LLMs) are required for both forward and backward propagation. 
    In short, MatMuls are the primary reason for the   exorbitant electricity costs of AI and for its adoption barrier by millions of businesses 
    (c.f.\ GPT-4 training estimated at \$300 million and the next model of OpenAI is projected to cost 100x more).  See Section~\ref{sec:matmul_in_ai} for an overview. Beyond AI, matrix multiplication is a central operation in many other applications, including vector databases, quantum simulations, graphics rendering and statistical physics to mention a few. 
    \item Efficiency: Our method for matrix multiply has minimal overhead over the direct matrix multiplication algorithm. Specifically, the running time of our algorithm is $N\cdot (1+o(1))$ where $N$ is a well-known baseline implementation of matrix multiplication (in practice, $N = O(n^3)$). We do not introduce any hidden large constant or significant terms. Thus, our protocol is not only theoretically interesting, but is also very promising from a practical perspective.

    \item Security: 
    Our protocol's soundness relies on a hardness assumption of a certain \emph{direct-product} problem of \emph{correlated} instances (essentially random low-rank linear equations).     
    Under this direct product conjecture, no adversary with comparable runtime can obtain non-trivial advantage which causes the verifier to accept with higher probability. In particular, multiplying the \bf all-zeroes \rm matrices $A=B= {\bf 0 \rm}^{n\times n}$ using our protocol is \emph{as hard as any other product}! 
\end{itemize}

\paragraph{MatMul algorithms.} 
The discovery of Fast Matrix Multiplication (FMM) algorithms by Strassen \cite{strassen1969gaussian} and subsequent improvements~(starting with~\cite{coppersmith1987matrix} and most recently~\cite{DBLP:conf/soda/WilliamsXXZ24}), allowing to multiply $n\times n$ matrices in sub-cubic time $n^\omega \ll n^3$, had a profound theoretical impact in computer science and algorithm design. Our PoW protocol can be instantiated and applied with any of the above FMMs as the baseline algorithm.
Unfortunately, FMM algorithms are unlikely to be practical on any imaginable hardware, as they are are highly asymptotic (earning them the infamous name ``galactic algorithms"), or require significant memory.  

Indeed, since our work is meant to serve as a practical alternative to Bitcoin, our analysis assumes 
\emph{na\"ive MatMul} as the most practically-efficient baseline algorithm for $n\times k$ and $k\times m$, which runs in $\Theta(nkm)$ time. 

We stress that our protocol's security analysis holds \emph{even in the presence of FMM algorithms}, assuming that they are used to attack the protocol itself. Developing an analogues PoUW scheme to \cupow\;  for Strassen-like FMM algorithms is an intriguing question from a theoretical standpoint, but this is not the focus of this paper. 


\paragraph{A formal definition and a framework.}
As an independent contribution, to the best of our knowledge, we formulate the first non-trivial definition for proofs of work associated with computational tasks. This definition allows reasoning about different schemes and compare them one to another. 

\subsection{Related Work}
There have been various proposals to completely get rid of the wastefulness of PoW mining by requiring the ``waste'' of a different resource. Most well known is the notion of ``proof of stake'' where participants in the network need to lock up a sum of money and thereby receive voting rights proportion to their quantity of holdings (e.g.,~\cite{DBLP:journals/sigmetrics/BentovLMR14,DBLP:journals/iacr/GiladHMVZ17,DBLP:conf/sp/KerberKKZ19} to mentioned just a few). Other resources, such as space~\cite{DBLP:conf/crypto/DziembowskiFKP15,DBLP:conf/fc/ParkKFGAP18,cohen2019chia}, have been studied and deployed in systems. While these avoid energy waste, they incur waste in other domains, and they all introduce new economic and cryptographic challenges and tradeoffs due to the nature of the other resource being utilized to secure the network. 

\paragraph{Trapdoor matrices and algorithmic speedup.} In two recent beautiful works by Vaikuntanathan and Zamir~\cite{VaikuntanathanZamir25} and Braverman and Newman~\cite{BravermanNewman25}, the authors discovered a technique for sampling pseudorandom matrices such that whomever has a secret key can multiply them in near-linear time. The suggested constructions rely on standard cryptographic assumptions, i.e., learning parity with noise (LPN). The authors of~\cite{VaikuntanathanZamir25} suggest to use those so-called ``trapdoor matrices'' in various algorithmic applications that rely on products of random matrices in order to obtain a significant speedup; among a series of applications, they emphasize a faster-than-known classification inference algorithm. 
The authors of~\cite{BravermanNewman25} suggest applications of a similar technique in a context of secure delegation of linear algebra, wherein a (weak) client wishes to compute some public function of its private data. They achieve such protocols for linear public functions (i.e., matrix-matrix multipllication or matrix-vector multiplication) by ``masking'' the client's private input with a ``trapdoor matrix''. 

The above goal is incomparable and orthogonal to proofs of work, and the techniques seem to be unrelated. First, the goal in proofs of work is somewhat reverse: we want hardness to be uniform for everyone, even for those who choose the inputs. Second, in proofs of work there is no possibility for hidden trapdoors as all of the computation is done locally by the party doing the computation. Whether ideas from the above works have any indirect bearing to our protocol is an interesting question.

\section{Overview of our PoUW}\label{sec:overview}
We first recall the main properties and syntax we would like to have in a PoUW, focusing on the task of matrix multiplication. In this overview, we ignore the economic value of multiplying two matrices, being a central operation in many large scale computations (notably AI training and inference; see Section~\ref{sec:matmul_in_ai}), and focus only on the technical aspects of a PoUW. A PoUW consists of two algorithms, $\mathsf{Solve}$ and $\mathsf{Verify}$, with the following properties, stated informally:

\begin{itemize}
    \item $\mathsf{Solve}(\seed,A,B)$ gets as input a  seed $\seed$ and two $n\times n$ matrices $A,B$ (the instance). The procedure has two outputs: (1) a matrix $C$ and (2) a proof $\pi$. 

    \textit{Prover efficiency}: The running time of $\mathsf{Solve}$ is not much more than just multiplying $A$ and $B$. That is, if $t(n)$ is the time it takes to multiply $A$ and $B$, then $\mathsf{Solve}$ runs in time $t(n) \cdot(1+o(1))$.

    \textit{Proof size}: The size of $\pi$ is $o(t(n))$.

    \item $\mathsf{Verify}(\seed,\pi)$ gets as input a  seed $\seed$ and a proof $\pi$, and it outputs a bit.
\end{itemize}

Correctness of a PoUW says that $C=A\cdot B$ (with probability 1) and also $\mathsf{Verify}(\seed,\pi)=1$ with probability~$\epsilon$ (over the randomness of $\mathsf{Solve}$ and over a uniformly random $\seed$), where $\epsilon>0$ is a parameter. Hardness of a PoUW says that no one can create a procedure that runs in time $t(n) \cdot(1+o(1))$ and succeed in creating $\pi$'s that are accepted with probability noticeably higher than $\epsilon$ (notice that this procedure does not need to even multiply matrices). That is, with the same computational budget there is no way to noticeably outperform the honest prover, $\mathsf{Solve}$.

We did not put any non-trivial restriction on the running time of $\mathsf{Verify}$ in order to emphasize the challenging parts of the notion of a PoUW. Our final scheme will be quite efficient also in communication and verifier complexity; we discuss this later.

\paragraph{The key idea.} The key idea underlying our PoUW is to encode noise into the input matrices, ensuring that any attempt to ``shortcut'' the computation would require effort comparable to performing the full computation. At the same time, we design an optimally efficient random self-reduction procedure, incurring negligible computational cost. The latter is obtained by leveraging, not only the task of producing an output, but more importantly, enforcing unpredictability in the transcript of the computation.

\paragraph{A MatMul random-self-reducibility-based PoW.} To gain intuition, we design a proof of work based on multiplying two arbitrary matrices that is \emph{not useful}. We will discuss later how we make it useful with minimal computational cost. The idea is to ask the prover to compute a ``noisy'' matrix product, instead of the original one. Specifically, we require computing the following:
\begin{align}
\label{eq:random-noise}
    C' = (A+E)\cdot (B+F),
\end{align}
where the matrices $E$ and $F$ are uniformly random  $n\times n$ matrices matrices. It is obviously necessary for $E$ and $F$ to be \emph{unpredictable} by the (malicious) prover, as otherwise it could pre-process the result of \eqref{eq:random-noise}. The latter property can be achieved by deriving $E$ and $F$ by applying a random oracle on $\sigma, A, B$. Since $\sigma$ is not known to the (malicious) prover ahead of time and $\mathcal O$ is a random oracle, the output of the random oracle is random and unpredictable. We can now turn this into a proof of (non-useful) work as follows:
\begin{itemize}
    \item The prover, on input seed $\sigma$ and two matrices $A,B$, does: 
    \begin{enumerate}
        \item Compute $E,F=\mathcal O(\sigma,A,B)$.
        \item\label{item:pow_matmul} Compute $C'=(A+E)\cdot (B+F)$.
        \item Compute $z= \mathcal O(C')$.
        \item Output $\pi = (A,B,z)$.
    \end{enumerate}
    \item The verifier, on input $\seed$ and $\pi = (A,B,z)$, does:
    \begin{enumerate}
        \item Recompute $z$ from $\sigma,A,B$ to check correctness. 
        \item Check if $z$ is below a threshold, say $2^{\lambda - \log(1/\epsilon)}$ (assuming~$\epsilon$ is the inverse of a power of~2). 
    \end{enumerate}
\end{itemize}

Since $E$ and $F$ are uniformly random matrices, the matrix $C'$ is uniformly random, no matter what $A$ and $B$ are. Therefore, $z$ will be below the above-mentioned threshold with probability exactly $\epsilon$, as needed. Suppose that $t(n)$ is the time spent by the prover on matrix multiplication in step~\ref{item:pow_matmul}. Note that $t(n) = O(n^3)$ if we use classical matrix multiplication, and $t(n) = n^\omega$ if we use fast matrix multiplication, where the current record is $\omega<2.371339$~\cite{abs-2404-16349}. Besides step~\ref{item:pow_matmul}, all operations are linear time in the input size and are thus negligible. Overall, the prover's runtime in the above system is $t(n)\cdot(1+o(1))$. 

For the above system to be a PoW, we need to argue that there is no way to perform noticeably better than the honest prover. In our system, the verifier checks consistency of $z$ and so the only way to perform better is if an attacker could somehow come up with $A$ and $B$ for which it can compute $z$ faster than doing a generic matrix multiplication algorithm. Suppose otherwise: we could use such a prover to compute the product of two \emph{random} matrices in time $o(t(n))$ --- a major breakthrough in algorithms. Assuming no such speedup is possible for random matrices, we conclude that a malicious prover must run in time at least~$t(n)$, as much as the honest prover one up to additive lower-order terms.

Obviously the above PoW is not useful since the prover does not output the product of $A$ and $B$. Instead, it can only output a noisy and seemingly useless value $C'=(A+E)\cdot (B+F)$. Naively, one can add another step to the prover's algorithm in which it computes the ``noise'' term $A\cdot F + E\cdot (B+F)$ and subtracts it from $C'$ to output $C=A\cdot B$. Computing this term, however, requires {two additional generic matrix products} (plus lower-order tasks).  Thus, the honest prover would need to perform $3$ matrix products, resulting in a prover whose running time is $t(n) \cdot (3+o(1))$, while a malicious prover still needs to do only a single matrix product for the same chance of convincing the verifier. Hence, the above scheme is not a PoUW.

\paragraph{Usefulness via transcript unpredictability.}
The issue with the above PoW not being useful was that ``peeling off'' the noise was too computationally expensive. Let us explore a different method of adding noise for which peeling it off is going to be relatively easy. Instead of sampling the noise matrices $E$ and $F$ as completely uniform, we sample them as low rank with parameter $r\ll n$. It is straightforward to do this by sampling two matrices $E_L,F_L$ of size $n\times r$ and two matrices $E_R,F_R$ of size $r\times n$, and then computing $E=E_L\cdot E_R$ and $F=F_L\cdot F_R$. As a result, $E$ and $F$ are two $n\times n$ random rank $r$ matrices. While $r$ is a generic parameter, for this overview, imagine that $r=n^{0.3}$ for convenience.

Notice that ``peeling off'' the noise in $C'$ is now computationally easy. Indeed, computing $A\cdot F + E\cdot (B+F)$ (cf.\ Eq.~\eqref{eq:random-noise}) can be done by two matrix product where in each of them, one of the matrices is low rank, allowing us to compute the product in time $O(n^2\cdot r)$. This term is asymptotically negligible compared to the cost of ``standard'' matrix multiplication, i.e., $O(n^3)$ using  classical matrix multiplication or $n^\omega$ using fast matrix multiplication. 

Once we have managed to get the correct output of the matrix product, the question is whether the new scheme is even a proof of work, or perhaps because $E$ and $F$ are ``less random'' an attacker can now convince the verifier with noticeably higher probability (and with the same computational effort) than the honest prover. The latter is unfortunately indeed the case: a malicious prover can choose easy to multiply matrices $A$ and $B$, say the all 0 matrices, and then its computational task boils down to merely multiplying $E$ and $F$. Being low-rank matrices, the adversary can compute $E\cdot F$ and then feed the output to the random oracle in time roughly $O(n^2\cdot r)$ which is much faster than the honest prover, requiring time $O(n^3)$ (or $n^\omega$ using fast matrix multiplication). It may seem as if we have not made any progress, but as we explain next, this is not the case.

As we concluded above, since $E$ and $F$ are low rank, we cannot utilize the output of the computation as a source of hardness. To this end, we make the following novel observation: \textbf{instead of using \emph{the output} of the computation as the proof of work, we utilize \emph{the transcript} of the computation}. This essentially ties the hands of the attacker and forces it to follow a prescribed algorithm, as we explain next. 

Consider the algorithm that multiplies $A'=(A+E)$ and $B'=(B+F)$ using a block-variant of the classical algorithm. Specifically, we split $A'$ and $B'$ into $r\times r$ blocks and initialize an all 0s matrix ${C'}^{(0)}$. We then iteratively compute the product of $A$ and $B$ via the following formula
        \[C'^{(\ell)}_{i,j}  := C'^{(\ell-1)}_{i,j} +  A_{i,\ell}\cdot B_{\ell,j},\]
where we treat the $(i,j)$-th index of each matrix as an $r\times r$ submatrix (or block). Namely, $i,j$, and $\ell$, all range in $1,\ldots,n/r$ (and we assume $r\mid n$ for simplicity). Obviously, $C'$ consists of the blocks $\{{C'}^{(n/r)}_{i,j}\}_{i,j\in [n/r]}$. 
        
This sequence  of $(n/r)^3$ matrices $\{{C'}^{(n/r)}_{i,j}\}_{i,j\in [n/r]}$ is referred to as the transcript of the matrix product. Our proof of work is then obtained by applying a random oracle on the transcript, instead of on the output. Specifically, we suggest the following PoUW:
\begin{itemize}
    \item The prover, on input seed $\sigma$ and two matrices $A,B$, does: 
    \begin{enumerate}
        \item\label{item:upow-noise} Compute $E_L,E_R,F_L,F_R=\mathcal O(\sigma,A,B)$ such that $E=E_L\cdot E_R$ and $F=F_L\cdot F_R$ and they are of dimension $n\times n$.
        \item\label{item:upow_matmul} Compute the transcript of $C'=(A+E)\cdot (B+F)$, denoted $\mathsf{tr} = \{{C'}^{(n/r)}_{i,j}\}_{i,j\in [n/r]}$. 
        \item\label{item:upow-ro} Compute $z= \mathcal O(\mathsf{tr})$.
        \item\label{item:upow-denoise} Compute $C=C'- (A\cdot F + E\cdot (B+F))$.
        \item Output $C$ and $\pi = (A,B,z)$.
    \end{enumerate}
    \item The verifier, on input $\seed$ and $\pi = (A,B,z)$, does:
    \begin{enumerate}
        \item Recompute $z$ from $\sigma,A,B$ to check correctness. 
        \item Check if $z$ is below a threshold, say $2^{\lambda - \log(1/\epsilon)}$ (assuming~$\epsilon$ is the inverse of a power of~2). 
    \end{enumerate}
\end{itemize}

As we have already explained, the above scheme satisfies correctness in the sense that a verifier will accept an honestly generated proof with probability $\epsilon$ and also $C$ is equal to $A\cdot B$. For efficiency, notice that steps~\ref{item:upow-noise} and~\ref{item:upow-denoise} can be implemented in time $O(n^2\cdot r)$ and step~\ref{item:upow-ro} requires hashing a string of length $(n/r)^3$ which takes time $o(n)$. Step~\ref{item:upow_matmul} requires computing the whole transcript of the matrix multiplication (which includes the output, $C'$). One way to compute it is directly running the classical matrix multiplication algorithm (in time $O(n^3)$). However, there is a more efficient way using fast matrix multiplication techniques. Indeed, every intermediate $r\times r$ block can be computed via a fast rectangular matrix multiplication that can be done non-trivially in some regimes of parameters. For instance, one can multiply a $n\times n^{0.3}$ matrix with a $n^{0.3}\times n$ matrix in time $n^{2+o(1)}$~\cite{DBLP:conf/soda/WilliamsXXZ24}. Thus, in some regimes one can implement step~\ref{item:upow_matmul} (and so step~\ref{item:upow-ro}) in time $O(n^3/r)$. Depending on $r$, let us denote the current state of the art runtime of the above algorithm for step~\ref{item:upow_matmul} by~$t_r(n)$. Since the latter is the dominant term, we get that the total runtime of the honest prover is~$t_r(n)\cdot (1+o(1))$.

For hardness, we conjecture that the algorithm we present for the honest prover is essentially optimal (up to low order optimizations) for an attacker whose goal is to cause the  verifier to accept. Indeed, an attacker can choose $A$ and $B$ and completely avoid running step~\ref{item:upow-denoise}, but we conjecture that there is no noticeable speedup possible in the computation of step~\ref{item:upow_matmul}.  

To this end, we formalize an assumption saying roughly that computing the transcript of a product of two random rank $r$ matrices requires time $t_r(n)$. With this conjecture, we conclude hardness of our scheme. We now give reasons to believe this conjecture, namely, why computing the transcript requires essentially to compute every $r\times r$ intermediate matrix independently. The main point is that every intermediate ${C'}^{(\ell)}_{i,j}$ essentially depends on the product of two $r\times r$ blocks from $A'=A+E$ and $B'=B+F$, denoted $A'_{i,\ell}$ and $B'_{\ell,j}$. Since $E$ and $F$ were chosen to be uniformly random low rank matrices, we can infer that both $A'_{i,\ell}$ and $B'_{\ell,j}$ are completely uniform, marginally. Thus, computing any single intermediate does require doing a single $r\times r$ random matrix multiplication, which as we conjectured cannot be sped up (from $O(r^3)$ directly or from $O(r^\omega)$ with fast matrix multiplication). Obviously the above argument does not work for computing all of the intermediates because they are correlated (remember that $E$ and $F$ are only rank $r$). However, recovering the correlations requires essentially solving a linear system of equations which, we conjecture, is as difficult as multiplying matrices. Thus, we believe our conjecture holds, unless significant algorithmic breakthrough is obtained.  

\begin{remark}[Memory complexity]
    As described the memory consumption of the PoUW is quite significant due to the need to store the whole transcript $\mathsf{tr}$. However, there are ways to improve this. For example, instead of hashing the whole transcript, one can hash each of the $(n/r)^3$ intermediate matrices separately, and use each one of them as a potential proof for solving the puzzle. The latter, in turn, causes the number of attempts one gets in solving the puzzle scale proportionally with the size of the matrices they multiply. This is a useful feature if the system needs to support varying matrix sizes. Additionally, note that in terms of hardness, a similar argument to the above applies in this case: each intermediate matrix is the result of the product of two independent, each marginally uniformly random matrices, causing the output to be ``computationally'' unpredictable (in a fine-grained sense).
\end{remark}

\begin{remark}[Improving verifier's complexity]
 The scheme as described above is relatively expensive on the verifier's side. Indeed, without any optimizations, the verifier needs to repeat the prover's work to get convinced that $\pi$ is ``valid''. Using standard tools we can improve this. Indeed, we can delegate the verifier's work (when $\mathsf{Verify}$ should accept) to the prover by asking the latter to compute a SNARK certifying that the verifier would have accepted the proof had it seen it. Furthermore, if privacy of $A,B$ is an issue, the prover can use a zkSNARK.
\end{remark}

\begin{remark}[Using SNARKs as a source of hardness?]
The above paragraph leads to a natural idea of using the SNARK computation itself as a source of hardness. Indeed, we could ``force'' the prover to follow a specific algorithm by asking it to generate a SNARK proof attesting to this fact. For this purpose one could avoid noising and denoising altogether if the SNARK computation itself is keyed (as in Micali's~\cite{DBLP:conf/focs/Micali94} random oracle-based SNARK).  This indeed forces the prover to work and prevents a malicious prover from doing shortcuts. The problem is that we typically think of $\epsilon$, the probability that $\mathsf{Solve}$ outputs an accepting proof, as being small. Nevertheless, the above approach requires the prover to compute a SNARK \emph{every time}, making the efficiency property of the scheme quite poor. (SNARKs require significant time to compute.) In our suggestion, on the other hand, only when the (rare) event that a valid $\pi$ is found, then we expect the solver to generate the (zk)SNARK, and so the cost of this computation is amortized and becomes negligible.
\end{remark}

\section{Open Problems and Future Directions}\label{sec_discussion}

The core idea in this paper (in particular, the PoUW Scheme in Algorithm \ref{fig:NoiseDeNoise2})
is performing PoW on \emph{intermediate} 
computations as opposed to the \emph{output} of the task $A\cdot B$. Once one subscribes to this idea and to a stronger security assumption (direct-product of random but correlated subproblems), it is natural to ask 
 how general is our approach: 

\begin{problem}
    What \emph{other} computational tasks $f(x)$ beyond MatMul admit an efficient ($1+o(1)$ multiplicative overhead) PoUW scheme? Can we characterize the properties of functions $f(x)$ for which an efficient PoUW scheme exists?
\end{problem}
Some natural and interesting candidates are solving linear equations, graph search and routing problems, and database problems like pattern matching and string problems (e.g., DNA sequencing). Some of these functions lack a natural  \emph{decomposable} structure into sub-problems $f(x) = \bigotimes g(x_i)$ like MatMul, so it is less clear how to use our notion of random self-reducibility on subproblems.

\paragraph{Alternative Noise schemes for MatMul PoUW}
Our PoUW scheme in Algorithm \ref{fig:NoiseDeNoise2} is merely 
one simple instance of an efficient (sub-cubic overhead) random self-reduction for MatMul, but there is in principle ``nothing holy" about using \emph{low-rank} noise -- Indeed, all we really need  the \emph{noise structure} to satisfy is that it is: (i) easy to add and remove ($o(n^3)$), (ii) induces \emph{unpredictable marginals}   
(random) on tiles, \emph{no matter} how $A,B$ are chosen. A natural question is whether there are even more efficient (and secure) schemes:  

\begin{problem}
Is there a better PoUW scheme for MatMul than Algorithm \ref{fig:NoiseDeNoise2}, either in terms of \emph{computational overhead} ($O(n^2r)$) 
or the \emph{security assumption} (ideally both)? 
\end{problem}

In Appendinx \ref{append_sec_rand_rotation} we present an alternative candidate PoUW scheme with \emph{self-canceling noise}, i.e., which \emph{preserves} the output $AB$ of the original (useful) problem, without any need for "denoising" altogether. This scheme has  
$O(n^2\log n)$ computational overhead (at least for certain matrix dimensions) instead of   
$O(n^2r)$ of Algorithm \ref{fig:NoiseDeNoise2}, but only applies to 
\emph{high-rank} matrices (typically satisfied by useful real-world matrices), and relies on a more subtle security assumption.  Nevertheless, it 
demonstrates the \emph{flexibility} of random self-reductions for MatMul, and the generality of our PoUW framework. 

\paragraph{A provably secure scheme.}
A fascinating question is whether one can design a PoUW scheme for matrix multiplication (or other ``useful'' operations) with provable security under more standard (albeit possibly strong) assumptions. For instance, is there a PoUW from general purpose program obfuscation assumptions?

\begin{problem}
Is there a PoUW scheme for useful operations from more standard or well-studied assumptions? 
\end{problem}

\section{Preliminaries}
\paragraph{Notation.}
For an integer $n\in \mathbb N$, we denote $[n]=\{1,\ldots,n\}$. For a set $D$, $x \leftarrow D$ denotes sampling an element from $D$ uniformly at
random. For a distribution $\mathcal D$ over a set $D$, $x\leftarrow \mathcal D$ denotes sampling an element of $D$ according
to $\mathcal D$. By $U_m$ we denote the uniform distribution over $\{0, 1\}^m$.

\paragraph{Random oracle model and complexity.}
Throughout this paper we assume the random oracle model (ROM). Proof of work protocols are easiest to define and prove secure in this model because it allows counting a query to the RO as a single operation/time unit.
Nevertheless, everything can be heuristically applied to the standard
model, by replacing the random oracle with a cryptographic hash function (say SHA-256 or BLAKE3).

In addition to supporting random oracle queries, we allow parties (as usual) to perform arithmetic operations. We assume an underlying field $\mathbb{F}_q$ that supports addition and multiplication. We count complexity of algorithm in our model by separately accounting for the number of random oracle queries they make and the number of field operations they do. 

\section{Definitional Framework for Proofs of Useful Work} 
\label{sec:pouw-def}

A proof of useful work (PoUW) is a method for a prover to convince a verifier that it performed sufficient computation to solve a particular computational task. Importantly, the task the prover chooses to solve is arbitrary, perhaps chosen by the prover itself, and does not come from any pre-specified distribution. The task may even remain private, depending on the application. Soundness requires that any (malicious) prover running sufficiently faster than the honest prover will not be able to fool the verifier to accept with high probability (even if they completely control the task they solve). More precisely, the ratio between the amount of computational effort and the probability of convincing a verifier is very close to the honest party's ratio). In other words, a proof of useful work has to satisfy the following two properties:
aim to achieve the following two desiderata simultaneously:
\begin{enumerate}
    \item \textbf{Usefulness}: Acceptable proofs can be generated by solving ``useful'' computational problems on \emph{arbitrary} inputs.  
    \item \textbf{Efficiency}: The (honest) prover needs to be almost as efficient as only solving the task (without proving).
    \item \textbf{Hardness}: Generating acceptable proofs is guaranteed to necessitate actual work, similar to the amount of work required honestly.
\end{enumerate} 

Classical proofs of work~\cite{DBLP:conf/crypto/DworkN92} only require the last property to hold, and have no associated useful computation. Usually, one can leverage classical proofs of work to generate \emph{random} instances of ``useful'' problems by sampling them from the proof of work challenge. In this way, for appropriate tasks (where average case complexity is well understood), a
prover is guaranteed to spend a certain amount of work in producing proofs. However, while the algorithm being executed for generating a proof of work is related to a concrete computational task, the instance itself is random and detached from any fixed
instance that someone may actually want to solve. Thus, such systems are not considered useful according to the above.

We formalize a proof of useful work as a protocol between a prover $P$ and a verifier $V$. At a high level, $P$ has an instance of a task that it needs to solve. Depending on the application, the instance could very well be already in $P$'s possession, or alternatively, may be given to $P$ from an external entity. Then, $P$ and $V$ engage in a protocol where eventually $P$ outputs a solution to the task and $V$ gets convinced that $P$ spent sufficient computational effort. 

It is very common to see definitions of proofs of work as non-interactive algorithms rather than interactive protocols. We may jump back and forth between the two methods, depending on the context. We remark that all of our protocols are public-coin and therefore we can use the Fiat-Shamir heuristic~\cite{FiatS86} to remove the interaction and prove security in the Random Oracle (RO) model . Thus, in what follows we assume the RO model, meaning that all parties have query access to a public random function. 

Since our proofs of work are associated with a useful computation, we assume a function $f\colon \{0,1\}^n\to\{0,1\}^n$, reflecting the computation being done during the proof of work. Above, we assume that $f$ has the same domain and range for simplicity, and all of our definitions extend to the case where the domain and range differ. Further, all of our definitions extend to the case where $f$ could have several ``legal'' outputs per input, namely, where $f$ is a relation. We assume that $f$ is a function for simplicity.
We assume that there is a canonical algorithm \textsc{Alg} for solving $f$, namely, the algorithm gets as input $\tsk\in \{0,1\}^n$ and outputs a solution $\sol\in\{0,1\}^{n}$ such that $f(\tsk)=\sol$. Denote $t(n,\tsk)$ the complexity of $\textsc{Alg}$ in solving $\tsk \in \{0,1\}^n$. We denote $t(n) = \max_{\tsk} f(n,\tsk)$ the worst-case complexity of $\textsc{Alg}$ over all inputs of size~$n$. 

Throughout we assume that $\lambda\in \mathbb{N}$ is a security parameter, and assume that $\{0,1\}^\lambda$ is the domain and range of the random oracle. One should think of  $n\gg \lambda$ where the relationship between the two is via some polynomial. Syntactically, the definition of a PoUW involves two oracle-aided algorithms:
\begin{enumerate}
    \item   $\mathsf{Solve}(\seed,\tsk)$: is a probabilistic algorithm that takes as input a seed $\seed\in\{0,1\}^\lambda$ and an instance $\tsk\in \{0,1\}^n$, and outputs a solution $\sol\in \{0,1\}^n$ and a proof $\pi\in\{0,1\}^*$.
    \item $\mathsf{Verify}(\seed, \pi)$: is a deterministic algorithm that takes a seed $\seed\in\{0,1\}^\lambda$ and a proof $\pi\in\{0,1\}^*$ as input, and outputs 0 (for failure) or 1 (for success).
\end{enumerate}

The role of $\seed$ is to guarantee
that the computational work spent in order to ``solve'' a task is ``fresh,'' i.e., executed when $\seed$ became publicly known to avoid preprocessing. One can further assume that there is a procedure $\mathsf{Keygen}$ that outputs $\seed$, and also possibly a $\mathsf{GenPP}$ procedure that outputs public parameters $\mathsf{pp}$ that all algorithms get as input. We simplify and avoid formally stating these for succinctness and clarity. 

We capture the efficiency, completeness, and soundness properties in the next definition. 

\begin{definition}[Proof of useful work]
\label{def:pow-relation}
An $\epsilon$-\underline{difficulty} and $\alpha$-\underline{overhead} PoUW for $f$ consists
of two algorithms $\mathsf{Solve}$ and $\mathsf{Verify}$, satisfying the efficiency, completeness, and hardness properties, described next.

\medskip
\noindent \textbf{Efficiency}: For any $\seed\in\{0,1\}^\lambda$ and any $\tsk\in \{0,1\}^n$:
\begin{itemize}
    \item  $\mathsf{Solve}(\seed,\tsk)$ runs in time $t(n)\cdot(1+\alpha(n))$.
    \item $\mathsf{Verify}(\seed,\pi)$ runs in time $c\cdot t(n)$ for constant $c\in (0,1)$ for any  possible output $\pi$ of $\mathsf{Solve}$.
\end{itemize}

\medskip
\noindent \textbf{Completeness}: For every $\lambda,n\in\mathbb N$, every $\seed\in\{0,1\}^\lambda$, and every $\tsk\in\{0,1\}^{n}$, it holds that 
$$\Pr[f(\tsk)=\sol]=1 \;\text{and}\; \Pr[\mathsf{Verify}(\seed,\pi)=1]\ge \epsilon(n),$$
where the probabilities are over the randomness of $\sol,\pi \leftarrow \mathsf{Solve}(\tsk,\seed)$ and over the random oracle.

\medskip
\noindent \textbf{Hardness}: 
There exists a constant $C>0$ such that for any algorithm $\mathsf{Solve}^\star$ running in time $t'= t(n)$ with $t'(n)\le t(n)$ with polynomial time preprocessing, it holds that
    $$\Pr\left[\mathsf{Verify}(\seed,\pi)=1\right]< \max \left \{C\cdot \frac{t'\cdot\epsilon(n)}{t}, 2^{-\lambda}\right \},$$
    where the probability is over the randomness of $\sigma\leftarrow \{0,1\}^\lambda$ sampled uniformly at random, computing $\pi \leftarrow \mathsf{Solve}^\star(\seed)$, and over the random oracle.

\end{definition}

One can introduce stronger notions of security, considering the success probability of a malicious prover when working on multiple instances in parallel (i.e., whether amortization exists), or when the adversary receives challenges in a streaming fashion and it only needs to show slight advantage over the naive process. The latter property is essentially what is needed for Bitcoin-style blockchains, and we formalize the security definition in Appendix~\ref{sec:poisson}.
     We conjecture that our PoW mechanism will satisfy all properties satisfied by the PoW of Bitcoin, and focus in the main body on the simpler definition to simplify presentation and convey the main new ideas.
We conclude this section with several remarks about the definition. 

\begin{remark}[Tunable difficulty]
In applications it is often necessary to be able to tune the hardness of the problem to solve. For example, in Bitcoin, the more energy the prover invests, the higher are the chances of convincing the verifier; moreover, once in a while the overall difficulty of mining becomes higher, requiring more energy to mine a block. 

This can be achieved in two ways. One way is to have $\epsilon$ fixed, say $0.1$, and vary the size of the tasks $\tsk$ the prover need to solve. Another way is to fix $n$ and tune $\epsilon$, requiring to solve ``more instances'' to achieve the same probability of convincing the verifier. 

In our final protocol, both $n$ and $\epsilon$ are tunable, and so one can either vary $n$ (choose smaller or larger instances to solver) or vary $\epsilon$ (require less or more probability of success) to tune the difficulty of the prover. In the above definition we treat $\epsilon$ as a fixed function only for simplicity.
\end{remark}

\begin{remark}[One vs.\ two-sided correctness]
    We require that the (honest) prover outputs the correct $\sol$ for each $\tsk$ it wants to solve (with probability 1). We do so for simplicity and because our construction will satisfy this strong requirement. However, it is plausible to consider weaker definition where, either a correct solution or an approximation thereof is outputted with reasonably high probability.
\end{remark}

\begin{remark}[Prover's efficiency]
    The efficiency of $\mathsf{Solve}$ is parametrized via a function $\alpha$ to which we refer as \emph{overhead}. Obviously, the overhead should be as small as possible, so as to incentivize running $\mathsf{Solve}$ instead of solving $f$ and trying to find $\pi$ separately. Indeed, the latter already gives a PoUW with constant multiplicative overhead (that depends on $\epsilon$). See Section~\ref{sec:genericPoWR} for details. Thus, the interesting regime is where $\alpha(t(n)) \ll t(n)$ and ideally $\alpha(n)=o(t(n))$. The latter is achieved by our protocol.
\end{remark}

\begin{remark}[Additional constraints]
    One could impose additional constraints on the efficiency of the protocol, e.g., that the proof $\pi$ sent from $P$ to $V$ is very short, or that the verifier's running time is sufficiently small. As we elaborate later in the paper, one can generically achieve these properties by using off-the-shelve tools such as succinct non-interactive arguments (SNARGs). (Note that computing the SNARG is only done when a ``valid'' $\pi$ is found, and so the complexity of computing the SNARG is amortized away.)
\end{remark}

\subsection{A (Trivial) PoUW for Any Function}
\label{sec:genericPoWR}
We show that a PoUW for any function $f$  with an associated $\textsc{Alg}$ running in time~$t(n)$ on instances of size~$n$ exists, albeit this scheme is not very useful since its associated $\alpha$ is large. Specifically, there is a gap between the complexity of $P$ and of the associated algorithm, and also there is a $P^\star$ that has significantly less complexity than $P$ and can generate ``valid'' $\pi$'s. At a very high level, the scheme is doing a computation plus proof of work, independently, as follows.

\begin{itemize}
    \item $\mathsf{Solve}$: on input $\seed$ and $\tsk$, the prover solves the task 
 by running $\textsc{Alg}$ on $\tsk$ to get $\sol$. Also, independently, it does a proof of work to convince the verifier that it spent enough time.  Specifically, for a  constant $c$ or perhaps slightly super-constant (see discussion below), it chooses $T=c \cdot t(n)$ nonces $\omega_1,\ldots,\omega_{T}$ for the random oracle, and computes $z_i = \mathcal O(\seed, \omega_i)$. It then sets $\pi$ to be the $\omega_i$ for which $z_i$ starts with the most $0$'s.

 
 \item $\mathsf{Verify}$: 
 on input $\seed$ and $\omega$, the verifier outputs~$1$ if and only if $\mathcal O(\seed, \omega)$ has $\log (t(n))$ leading~$0$'s.  

 The verifier's running time is thus $\tilde O(1)$.
\end{itemize}

\begin{claim}[A PoUW for $f$]
    The above PoUW protocol with overhead $\alpha=c$ and difficulty $\epsilon=1-e^{-c}$.
\end{claim}
\begin{proof}
    Since $\mathcal{O}$ is a random oracle, its output on previously-unqueried  points is uniformly distributed in $\{0,1\}^\lambda$. Overhead is $\alpha=c$ directly by construction. Furthermore, the probability of success for a single query is $p = 1/t(n)$. The probability that all $T$ attempts fail, is $(1-p)^{c\cdot t(n)} \le e^{-c}$, so completeness holds. Hardness follows  from the random oracle property: A method for ``guessing'' a nonce that leads to an output that has $\log T$ leading 0s amounts to ``guessing'' the output of a random oracle without even making the appropriate query.
\end{proof}

    The protocol is not useful in applications: either $P$ has linear overhead over the complexity of \textsc{Alg}, or an adversarial $P^\star$ can gain linear advantage in running time over $P$.
    Consider the following malicious prover $P^\star$: $P^\star$ does not multiply any two matrices, and instead only finds the good nonce out of $T$ arbitrary ones. Clearly, this $P^\star$ has exactly the same probability of convincing $V$ as $P$, but it is more efficient. Concretely, in this protocol protocol  $t_{P^\star}(n)=c\cdot t(n)$ which is a factor $(1+1/c)$ faster than $t_{P}(n)$. It is possible to make the advantage of $P^\star$ sub-linear by setting $c=\omega(1)$, but this takes the other issue to the extreme: the honest prover $P$'s complexity is a multiplicative factor $(1+c)$ larger than the complexity of $\textsc{Alg}$. 
\section{A PoUW for Matrix Multiplication}

In this section we present a PoUW protocol for the matrix multiplication task. In this context, the task consists of two $n\times n$ matrices $A$ and $B$ with entries integers in $\mathbb F_q$ (we assume field operations), and the solution is an $n\times n$ matrix $C$ being the product of $A$ and $B$ with entries also being in $\mathbb F_q$.\footnote{The protocol can be extended to rectangular matrix multiplication, but again we favor simplicity of presentation.} 

This section is structured as follows. We first present the canonical matrix multiplication algorithm that we assume (referred to as \MatMul). Our protocols are roughly based on the random self reducibility property of matrix multiplication. Our first protocol is rather straightforward, and while being better than the generic one in Section~\ref{sec:genericPoWR}, it is still not an optimal PoUW. Our second protocol combines the idea of random self reducibility of matrix multiplication together with an additional structural property of the algorithm we use. This protocol, as we show, is an optimal PoUW under a new algorithmic conjecture. 

\subsection{The Canonical \MatMul Algorithm}\label{sec:canonical_matmul}
 For concreteness, we describe the textbook algorithm for matrix multiplication, denoted \MatMul.  In this algorithm, the $(i,j)$ entry of the output is computed by performing an inner product  between the $i$-th row in the left matrix and the $j$-th column in the right matrix. We generalize this algorithm to operate on $r\times r$-size blocks (or so-called ``tiles''), instead of on single entries.  See Algorithm~\ref{fig:matmul} for the pseudocode. 

\begin{myalg}[Canonical Matrix Multiplication Algorithm]\label{fig:matmul}
\noindent\underline{$\MatMul_r(A,B)$}:   \quad $\sslash$ $A\in {\mathbb F}_q^{n\times k}, \quad B\in {\mathbb F}_q^{k\times m}, \quad r\in [n], \quad r | n \text{ and } k \text{ and } m$  

    \begin{enumerate}
        \item Initialize an all-$0$ matrix $C^{(0)}\in [-q, q]^{n\times m}$.
        \item Partition $A$, $B$, and $C$ into blocks of size $r\times r$. The $(i,j)$th block of a matrix $X\in \{A,B,C\}$ is denoted $X_{i,j}$.
        
        \item For each $i\in [n/r]$, $j\in[m/r]$, and $\ell\in[k/r]$, compute 
        \[C^{(\ell)}_{i,j}  := C^{(\ell-1)}_{i,j} +  A_{i,\ell}\cdot B_{\ell,j},\]
        where $A_{i,\ell}\cdot B_{\ell,j}$ can be computed via any classical matrix multiplication algorithm.
        \item  Output $C^{(n)}$.
    \end{enumerate}
\end{myalg}

The complexity of $\MatMul_r$, when multiplying an $n\cdot k$ matrix with a $k\times m$ matrix,  is governed by the computation of the  $nkm/r^3$ ``intermediate' values.  When we implement the computation of the intermediate values naively, we obtain an algorithm with  complexity 
     $t_\MatMul(n,k,m) = O(n\cdot k\cdot m)$.
To simplify notation, when we multiply two $n\times n$ matrices, i.e., $n=k=m$, we write $t_\MatMul(n) = t_\MatMul(n,k,m)$. We assume throughout that $r$ divides $n,k,m$.

\paragraph{$\MatMul_r$ transcript.} 
The transcript of a $\MatMul_r$ (Algorithm~\ref{fig:matmul}) execution consists of all $nmk/r^3$ intermediate $r\times r$ matrices $\{C^{(k)}_{i,j}\}_{i\in [n/r],j\in [m/r],\ell\in [k/r]}$ during the computation of $A\cdot B$ for $A\in {\mathbb F}_q^{n\times k}$ and $B\in {\mathbb F}_q^{k\times m}$.  
\begin{definition}[Transcript]
    The transcript of $\MatMul_r(A,B)$ consists of 
    \[\Tr(\MatMul_r,A,B) = \left\{C^{(k)}_{i,j}\right\}_{i\in [n/r],j\in [m/r],\ell\in [k/r]}.\]
\end{definition}

Evidently, in general, there are $nmk/r^3$ intermediate matrices, each of which is of size $r^2$. We give some examples, assuming that  $n=m=k$ for simplicity. If  $r=n$, there is exactly 1 intermediate state and it is the output of the computation. Alternatively, if $r=n/2$, there are 8 intermediate states, two states per one of the 4 sub-rectangles of $C$. Lastly, if $r=1$, then there are $n^3$ intermediate matrices, each being a single value in ${\mathbb F}_q$.

\paragraph{Algorithms for computing the intermediates.}
As mentioned, one can compute all of the $nkm/r^3$ intermediate matrices using a direct matrix multiplication algorithm, requiring $O(r^3)$ work per intermediate matrix, resulting with overall complexity $O(nmk)$. However, using fast matrix multiplication techniques one can do better. 

The first idea is to use fast square matrix multiplication algorithms. Specifically, every intermediate value is obtained by adding to a previously calculated intermediate value the product of two $r\times r$ matrices. Since this can be done in time $r^\omega$ where $\omega$ is the exponent  of matrix multiplication (currently standing at $\omega = 2.371552$~\cite{DBLP:conf/soda/WilliamsXXZ24}), we obtain an algorithm with overall complexity $O(nmk/r^{3-\omega})$, beating the naive one.

A somewhat more efficient approach in many cases is to utilize fast rectangular matrix multiplication algorithms. Specifically, consider the matrix induced by the values $C_{i,j}^\ell$ for a fixed $\ell$. Let us focus on $\ell=1$ for concreteness. This matrix is, in fact, the product of the left-most column-strip of $A$ (when we view the matrix as having $r\times r$ atomic blocks as elements) and the top-most row-strip of $B$ (essentially computing the outer product of these vectors). We shall denote $\omega_{m}$ the exponent 
 of the best matrix multiplication algorithm for multiplying $n\times m$ and $m\times n$ matrices for $m<n$. We get an algorithm that runs in time $(n/r)\cdot n^{\omega_r} = n^{\omega_r+1}/r$,  a term that could be $ o(n^3)$.  As a concrete setting, the current record says that computing an $n\times n^k$ times an $n^k\times n$ matrix for all $k\le  0.321334$ can be done in time $n^{2+o(1)}$~\cite{DBLP:conf/soda/Gall24,DBLP:conf/soda/WilliamsXXZ24}. That is $\omega_k = 2$ for all $k\le  0.321334$. For $r=n^{0.3}$, this results with an algorithm that works in time $n^{2.7+o(1)}$ for computing all of the intermediates.

To conclude, there are several ways to compute the intermediate values, some of which are asymptotically more efficient than the direct method. However, in practice the most efficient algorithm for reasonable values of $n$ is still the direct one because the other ones have either large hidden constants or require significant resources other than compute (e.g., much more memory accesses or an inherent sequentiality).

\subsection{A Transcript Unpredictable Encoding Scheme} \label{sec:tue}
Our template relies on a generic invertible operation we call a transcript unpredictable encoding scheme. The goal of the encoding procedure of the scheme is to inject noise and perturb the input matrices so that the  transcript of the perturbed matrices has non-trivial entropy. The inverse operation allows to ``peel off'' the injected noise and recover the original output. Both operations are parametrized with $r\in [n]$, the same parameter that appears in $\MatMul_r$. The scheme consists of two deterministic algorithms $(\Noise,\DeNoise)$ with the following syntax:
\begin{itemize}
    \item $A',B'\Leftarrow\Noise_r(\seed,A,B)$: On input a seed $\seed$ and two matrices $A$ and $B$, it outputs two matrices~$A'$ and~$B'$.
    \item $C\Leftarrow \DeNoise_r(\seed,C')$: On input a seed $\seed$ and a matrix $C'$, the procedure outputs a matrix $C$.
\end{itemize}

We now formalize the properties of a transcript unpredictable encoding scheme $(\Noise,\DeNoise)$.  Correctness says that one can ``peel off'' the noise that is added during the $\Noise$ operation and recover the original product of $A$ and $B$. The other (novel) feature, called \emph{transcript unpredictability}, says that the only way to learn a piece of the transcript of the computation is to directly compute it. Satisfying each of the two properties separately is trivial, and the challenge is of course to have them hold at the same time.

\begin{definition}[Transcript-unpredictable encoding]\label{def:tasnscript_unpredictable_encoding}
    A transcript unpredictable encoding scheme $(\Noise,\allowbreak\DeNoise)$ satisfies the following properties:

    \medskip\textbf{Completeness}:
    For every two matrices $A$ and $B$, and auxiliary input $r$, it holds that 
    \begin{align*}
        \Pr_{\seed}[\DeNoise_r(\seed, A'\cdot B') = A\cdot B \;\colon\; A',B' = \Noise_r(\seed,A,B) ]=1.
    \end{align*}

\medskip\textbf{$\epsilon$-transcript unpredictability}:
    There is a constant $C>0$ such that for every two $n\times n$ matrices $A,B$, and every  $r$, every \textbf{iterative numerical} algorithm $\mathcal A$ that runs in time $t=t(n)$ cannot compute the transcript of the product of the noisy matrices $A',B'$, except with probability~$C\cdot t/t_\MatMul\cdot \epsilon$. That is,
\begin{align*}
\Pr\left[ 
\begin{aligned}
&\forall i,j,k\in [n/r] \colon  \\ & \quad C^{(k)}_{i,j} = Z^{(k)}_{i,j}
\end{aligned} \,\middle\vert\, 
\begin{aligned}
    &A',B'=\Noise_r(\seed,A,B)  \\
    & \{C^{(k)}_{i,j}\}_{i,j,k\in [n/r]} = \Tr(\MatMul_r,A',B')\\
    &\mathcal \{Z^{(k)}_{i,j}\}_{(i,j,k)\in [n/r]} = \mathcal A(A',B')\\
\end{aligned}
\right] \le \frac{C\cdot t(n)}{t_\MatMul(n)}\cdot\epsilon(n),
\end{align*}
where the probability is over the choice of $\sigma\leftarrow \{0,1\}^\lambda$ and the internal randomness of $\mathcal A$.
\end{definition}

We refer to Appendix~\ref{appx:ue} for a more general abstraction (that applies to any function, and not only matrix multiplication), termed \emph{unpredictable encodings}, that we believe is of independent interest.

\subsection{From Transcript Unpredictability to a PoUW} 
An implication of transcript unpredictability is that computing the whole transcript requires (essentially) to perform the direct computation in time $t_\MatMul(n)$. This directly gives us a PoUW, as described next.

\begin{myalg}[A PoUW for Matrix Multiplication with Parameter $r$]\label{fig:functionalmatmul}
   \noindent\underline{$\mathsf{Solve}(\seed,A,B)$}:  \quad $\sslash$ $A,B\in {\mathbb F}_q^{n\times n}$

\begin{enumerate}
    \item Compute $A',B'=\Noise_r(\seed,A,B)$.
    \item 
    Computes $ C'=\MatMul_r(A', B') $. 
    
    Denote $z=\{C^{(k)}_{i,j}\}_{i,j,k\in [n/r]}$ the intermediate $r\times r$ matrices. 
\item Compute $C := \DeNoise_r(\seed,C')$.

\item Output $(C,\pi)$ where $\pi=(A,B,z)$.
\end{enumerate}
\noindent\underline{$\mathsf{Verify}(\seed,\pi)$}:
\begin{enumerate}   
\item Parse $\pi = (A,B,z)$. 
\item Recomputes from $\seed, A,B$ the correct value of $z$, and output~$1$ if and only if $z$ is correct.
    \end{enumerate}
\end{myalg}

\paragraph{Correctness.}
It is straightforward to verify correctness. Specifically, the output of $P$ is always $C=A\cdot B$ directly by the correctness of the algorithms $\Noise$ and $\DeNoise$. Also, the verifier accepts because it essentially repeats the computation of the honest $P$, necessarily ending up with the same $z$.

\paragraph{Hardness.}
 We now argue that any algorithm that runs in time strictly less than $t_\MatMul(n)$, cannot convince the verifier to output 1, except with small probability. Suppose that there is a $P^\star$ that runs in time less than~$t_{\MatMul}(n)$. The claim follows directly form the $\epsilon$-transcript unpredictability property of the $(\Noise,\DeNoise)$ operations, concluding that with probability at most $\epsilon$ such a $P^*$ will fool the verifier, concluding the claim.

\subsection{A First Instantiation}
Our first instantiation of the $\Noise$ and $\DeNoise$ operations is quite standard and relies on the classical random self-reducibility of matrix multiplication.  In this instantiation we use $r=1$.

\begin{myalg}[First Implementation of $\Noise$ and $\DeNoise$]\label{fig:NoiseDeNoise1}
    \noindent \underline{$\Noise_r(\seed,A,B) $}: $\quad \sslash  \quad A,B\in {\mathbb F}_q^{n\times n}$
    \begin{enumerate}
        \item Parse $\seed = E,F$,
        where $E,F \in {\mathbb F}_q^{n\times n}$.
        \item Output $A' = A+E$ and $B'= B+F$.
    \end{enumerate}

    \noindent \underline{$\DeNoise_r(\seed,C') $}: $\quad \sslash  \quad C'\in {\mathbb F}_q^{n\times n}$
    \begin{enumerate}
    \item Parse $\seed = E,F$,
        where $E,F \in {\mathbb F}_q^{n\times n}$.
    \item Compute $C'' = A\cdot F + E\cdot (B+F)$ (using two invocations of $\MatMul_r$ and a matrix addition).
    \item Output $C=C'-C''$.
    \end{enumerate}
\end{myalg}

\medskip
\noindent It is obvious that completeness holds as $$C' - C'' = (A+E)\cdot (B+F) - A\cdot F + E\cdot (B+F)  = A\cdot B.$$  
(notice that the $E$ and $F$ matrices are consistent between the two procedures).
We now discuss transcript unpredictability. Since $r=1$, the transcript of the computation consists of only the final output $C'$. For trace unpredictability, we essentially assume that matrix multiplication for random matrices requires $\Omega(n^\omega)$ time where $\omega$ is the matrix multiplication exponent. We formalize this next.
\begin{assumption}
    There is no algorithm to compute the product of  two random $n\times n$  matrices in time  $o(n^{\omega})$.
\end{assumption}

\paragraph{Efficiency of PoUW.}
The obvious downside is that a malicious solver can skip the $\DeNoise$ operation altogether and thereby save two matrix multiplications. In other words, the honest solver does at least three times the work it should have done only for computing $\MatMul_r$, i.e., the protocol's overhead is $\alpha\ge 3$. On the other hand, there is a $\mathsf{Solve}^*$ that does essentially only $t_{\MatMul}(n)$ work, concluding that the protocol is not useful.

\subsection{A Second Instantiation}
Our second instantiation of the $\Noise$ and $\DeNoise$ operations also relies on the classical random self-reducibility of matrix multiplication, however, here the unpredictability assumption is more novel.  In this instantiation we use $r$ as a parameter and explain how it can be instantiated below.

\begin{myalg}[Second Implementation of $\Noise$ and $\DeNoise$]\label{fig:NoiseDeNoise2}
    \noindent \underline{$\Noise_r(\seed,A,B) $}: $\quad \sslash \quad\sigma\in \{0,1\}^\lambda, \quad A,B\in {\mathbb F}_q^{n\times n}$
    \begin{enumerate}
        \item Parse $\seed = E_L,E_R,F_L,F_R$ as four matrices, where $E_{L},F_{L} \in {\mathbb F}_q^{n\times r}$ and $E_{R},F_{R} \in {\mathbb F}_q^{r\times \cd}$.
        \item Compute 
        $E = E_L\cdot E_R \text{ and } F = F_L\cdot F_R$ (invoking $\MatMul$ twice).
        \item Output $A' = A+E$ and $B'= B+F$.
    \end{enumerate}

    \noindent \underline{$\DeNoise_r(\seed,C') $}: $\quad \sslash \quad\sigma\in \{0,1\}^\lambda, \quad C'\in {\mathbb F}_q^{n\times n}$
    \begin{enumerate}
    \item Parse $\seed = E_L,E_R,F_L,F_R$ as four matrices, where $E_{L},F_{L} \in {\mathbb F}_q^{n\times r}$ and $E_{R},F_{R} \in {\mathbb F}_q^{r\times \cd}$.
    \item Compute $C'' = (A\cdot F_L)\cdot F_R + E_L\cdot (E_R\cdot (B+F_L\cdot F_R))$  (invoking $\MatMul$ 5 times).
    \item Output $C=C'-C''$.
    \end{enumerate}
\end{myalg}

\medskip
\noindent It is obvious that completeness holds as 
$$C'' = A\cdot F + E\cdot (B+F)$$ and therefore  
$$C' - C'' = (A+E)\cdot (B+F) - A\cdot F + E\cdot (B+F)  = A\cdot B.$$  
(notice that the $E$ and $F$ matrices are consistent between the two procedures).
We now discuss transcript unpredictability. Here, $r$ is a parameter and so the transcript of the computation consists of $(n/r)^3$ intermediate $r\times r$ matrices.  Further, $E$ and $F$ are completely random conditioned on being rank $r$ (see Section~\ref{sec:low-rank-lemmas} for a proof). The question is whether it is possible to compute the transcript of $\MatMul_r(A,B)$ without essentially performing the computation as we described  in Section~\ref{sec:canonical_matmul}. We formalize this next.
\begin{assumption}\label{conj_DP_LR_scheme}
    There is no algorithm to compute all the intermediate values when multiplying two random $n\times n$  rank-$r$ matrices running in time  $o(n^{\omega_r+1}/r)$.
\end{assumption}

Recall that the algorithm that we described in Section~\ref{sec:canonical_matmul} can be used to multiply any two matrices in $O(n^{\omega_r+1}/r)$ time, and in particular, to multiply random rank $r$ ones. One can wonder if there is a way to get a speedup by utilizing the low rank property. Suppose that the matrices one wants to multiply are denoted $A$ and $B$, and that the low rank decomposition of them is $A=A_L\cdot A_R$ and $B=B_L\cdot B_R$, where $A_L,B_L$ are $n\times r$ and $A_L,B_L$ are $r\times n$. Now, one can generate all the intermediates by computing a partial sum  $T_z=\sum_{i=1}^{z} A_R\cdot B_L$ in blocks of size $r$ (so $z$ ranges in $1,\ldots,n/r$), and then multiplying $A_L\cdot T_z \cdot B_R$. Each of the latter can be done by fast rectangular matrix multiplication, costing $n^{\omega_r}$, but there are $n/r$ different values of $z$. So, overall, we obtain the same running time of $O(n^{\omega_r+1}/r)$, as the general purpose algorithm. Using the state of the art rectangular fast matrix multiplication, we can see that the above assumption is true (because the number bits needed to write down all  intermediates -- $\Theta(n^3/r)$ -- is always a lower bound on the running time and $\omega_r=2$ for sufficiently small $r$.

\paragraph{Efficiency of PoUW.}
Here, as opposed to all of our previous suggestions, the PoUW is useful. Indeed, the honest solver $\mathsf{Solve}$ is doing exactly one product of two $n\times n$ matrices, as all of the products in $\Noise$ and $\DeNoise$ are products of $n\times n$ and $n\times r$ matrices, that cost $O(n^2\cdot r)$. Overall, by a direct analysis, the overhead is
$$ \alpha(n) = O(n^2 + n^2 \cdot r)/t_{\MatMul}(n).$$
Note that $\alpha(n)=o(1)$ (unless fast matrix multiplication is used). By our assumption on the fastest way to compute the intermediate matrices (see Section~\ref{sec:canonical_matmul}, there is no algorithm that runs in time $o(t_{\MatMul}(n))$, making our protocol a truly useful PoW.

\subsubsection{Properties of our Noise Matrices}\label{sec:low-rank-lemmas}
The method of generating noise matrices described in Algorithm~\ref{fig:NoiseDeNoise2} results with a uniformly random rank~$r$ matrix. Let $\mathcal E_{r,n}$ be the set of all $n\times n$ rank $r$ matrices. 
\begin{lemma}
    In the random oracle model and assuming that $\seed$ is unpredictable, then the induced distribution of $E$ and $F$ is uniformly random from $\mathcal E_{r,n}$ (with very small statistical error).  In particular, with very high probability, every $r\times r$ submatrix of $E$ and $F$ is marginally uniform.
\end{lemma}
\begin{proof}
   We show that $E$ is a uniformly random rank $r$ matrix of dimension $n\times \cd$. The analog statement for $F$ is proven in the same way.

 Let \( \mathcal{A}_{n,r} \) denote the set of all \( n \times r \) matrices of rank \( r \), and let \( \mathcal{B}_{r,\cd} \) denote the set of all \( r \times \cd \) matrices of rank \( r \).    First, in Claim~\ref{claim:random-is-full-rank}, we argue that except with small probability of error, $E_L$ is distributed uniformly at random in  \(\mathcal{A}_{n,r}\). The analog claim holds  for $E_R$ as well. We thus condition on $E_L$ and $E_R$ being uniformly random in their respective domains conditioned on being full rank.

   Let \( \mathcal{M}_{n,\cd,r} \subset {\mathbb F}_q^{n \times \cd} \) denote the set of all \( n \times \cd \) matrices over \( {\mathbb F}_q \) with rank exactly \( r \).
 By construction, the row space of \(E \) is the row space of \( E_L \), and since \( E_L \) is chosen uniformly from \( \mathcal{A}_{n,r} \), the row space of \( E \) is uniformly distributed over all \( r \)-dimensional subspaces of \( {\mathbb F}_q^n \). Similarly, the column space of \( E \) is the column space of \( E_R \), and since \( E_R \) is chosen uniformly from \( \mathcal{B}_{r,\cd} \), the column space of \( E \) is uniformly distributed over all \( r \)-dimensional subspaces of \( {\mathbb F}_q^\cd \).
Overall, the row and column spaces of \( E \) are independent and uniformly distributed over all possible \( r \)-dimensional subspaces of \( {\mathbb F}_q^n \) and \( {\mathbb F}_q^{\cd} \), respectively. This means (a) every matrix in \( \mathcal{M}_{n,\cd,r} \) can be produced and moreover each one is equally likely to be produced.

\end{proof}

\begin{claim}\label{claim:random-is-full-rank}
    Let \( A \) be an \( n \times d \) matrix (with $d<n$) whose entries are chosen independently and uniformly from $\mathbb Z_q$. Then, 
\[
P(\text{rank}(A) = d) = 1 - O\left( {1/q^{n-d+1}}\right).\]
\end{claim}
\begin{proof}
The matrix \( A \) will have rank \( d \) if and only if its \( d \) columns are linearly independent. For \( A \) to be rank-deficient (rank \( < d \)), there must exist a non-trivial linear combination of the columns that results in the zero vector. This is equivalent to saying that one of the columns lies in the span of the preceding columns.

The first column of \( A \) is non-zero with probability \( 1 - \frac{1}{q^n} \), since it must not lie in the zero subspace.
Given that the first column is non-zero, the second column is linearly independent of the first with probability \( 1 - \frac{1}{q^{n-1}} \), as it must not lie in the one-dimensional subspace spanned by the first column.
More generally, for the \( k \)-th column to be linearly independent of the previous \( k - 1 \) columns, it must avoid the \( (k-1) \)-dimensional subspace spanned by those columns. This occurs with probability \( 1 - \frac{1}{q^{n-k+1}} \).

The probability that all \( d \) columns are linearly independent is the product of these probabilities, giving that
\[
\Pr(\text{rank}(A) = \cd) = \prod_{k=0}^{d-1} \left( 1 - \frac{1}{q^{n-k}} \right) \ge \prod_{k=0}^{d-1} \exp\left({-\frac{1}{q^{n-k}}}\right) = \exp\left({-\sum_{k=0}^{d-1} \frac{1}{q^{n-k}}}\right),\]
where the inequality holds since 
$ e^{-x} < 1 - x$ for  \( x \in [0, 1] \). Further, since \( \sum_{k=0}^{d-1} \frac{1}{q^{n-k}} \le 2/q^{n-d+1} \) and since \(1 - x + \frac{x^2}{2} < e^{-x}\) for  \( x \in [0, 1] \), we conclude that
$
\Pr(\text{rank}(A) = n) \ge \exp\left({- {2}/{q^{n-d+1}} }\right) \ge 1 -  2/q^{n-d+1} .
$

\end{proof}

\section{
Using Native AI Workloads for Blockchain Consensus}
\label{sec:matmul_in_ai}

Artificial Intelligence (AI) tasks, both training and inference phases, can be conceptualized as sequences of matrix multiplication operations interspersed with various other operations such as activation functions, normalizations, and rounding. The matrix multiplications are typically the most computationally intensive parts of any existing networks, with complexities often scaling as cubic in contrast to the square complexity of many of the other operations. The operations performed on these matrices during training and inference can be broken down into a series of matrix multiplications.

\paragraph{Training phase.}
During the training phase, the model adjusts its weights based on the input data and the error of its predictions. The key steps involve:

\begin{itemize}
    \item Forward Propagation:

\begin{itemize}
    \item 
    Linear Transformation: Each layer performs a linear transformation, which is essentially a matrix multiplication between the input data matrix $X$
 and the weight matrix $W$. This layer computes $Z=XW$.
    \item Activation Function: After the linear transformation, an activation function $f$ (e.g., ReLU, Sigmoid) is applied element-wise to introduce non-linearity. 
\end{itemize}    
\item Backward Propagation:

\begin{itemize} 
\item Gradient Calculation: The gradients of the loss function with respect to each weight matrix are computed. This involves several matrix multiplications, particularly when applying the chain rule of differentiation through the layers of the network.
\item Weight Update: The weight matrices are updated using the gradients and a learning rate. The update step is also a matrix operation, albeit simpler than multiplication.
\end{itemize}
\end{itemize}

\paragraph{Inference phase.} During inference, the model uses the learned weights to make predictions. This primarily involves:

\begin{itemize}
    \item Linear Transformation and Activation: Similar to the forward propagation in training, the input data passes through the network layers, each performing a matrix multiplication followed by an activation function.
\item Output Layer: The final layer often includes a matrix multiplication to produce the output predictions.
\end{itemize}

\paragraph{Reducing $\cupow$ overhead in AI workloads.}
Utilizing the structure and predictability of MatMuls in AI workloads can provide significant speedups in practice (i.e., reduction in overhead). We briefly discuss two ideas. First, in several applications, notably inference, the weight matrices are known and are public at the start of the process, and so their hash can be preprocessed, avoiding the need to hash them at every layer ``on the fly.'' Second, in training applications, the matrices in a given level are typically the output of a product from the previous level (after activations). Thus, in principle, one need not commit to the matrices from scratch in every level, but rather one can utilize a commitment to the very first level and then ``recompute'' the needed matrices.  

\paragraph{Computational complexity and load.}
The computational burden in AI models primarily arises from matrix multiplications. For matrices of size $n\times k$ and $k\times m$, the multiplication has a complexity of $n k m$. When $n=k=m$, this simplifies to $O(n^3)$. Other operations, including activation functions, normalizations, rounding, etc. typically scale with $O(n^2)$ as they operate element-wise. 

Based on experiments and benchmarks we have performed, we estimate that matrix multiplications dominate the computational workload in neural network training. In particular, we find that in typical scenarios at least 50\% and often even 70-80\% of the total computational load is devoted to matrix multiplication.

\section*{Acknowledgements}
We thank Yonatan Sompolinsky for multiple discussions on the topic of this work and for valuable comments on the manuscript. We thank Ohad Klein for valuable feedback and particularly for suggesting the FMM-based algorithm for computing the intermediates. We thank Eylon Yogev and Mark Zhandry for useful remarks and suggestions on this manuscript.

\bibliography{refs}
\bibliographystyle{alpha}

\appendix

\section{A Self-Canceling Noise Scheme}\label{append_sec_rand_rotation}

There is in principle nothing ``holy'' about the low-rank structure of the noise matrices used in the \cupow\; scheme (Algorithm \ref{fig:NoiseDeNoise2}) -- Indeed, all we seem to need is a \emph{low-entropy} perturbation operation $(A,B)\rightarrow (A', B')$, whose \emph{structure} can be exploited for decoding the \emph{output} $AB$, but not for intermediate computations -- i.e., the perturbed matrices $A',B'$ should have \emph{marginally uniform} projections on $r\times r$ tiles (which is the hard distribution for MatMul). 

\paragraph{Self-Canceling Noise  via Pseudorandom Rotations} 
We observe that using \emph{pseudorandom rotations} on high-rank matrices enables to \emph{avoid the decoding (``noise peeling") step in Algorithm \ref{fig:NoiseDeNoise2} altogether}. The idea is to \emph{randomly rotate} $A,B$ using a fast \emph{pseudorandom orthonormal} matrix $R$, e.g., the \emph{Randomized Hadamard} Transform \cite{AC06} which allows $O(n^2 \log n)$-time rotations $AR$ using the Fast Fourier Transform (FFT),   
and run the same tile-based PoW scheme from Algorithm \ref{fig:NoiseDeNoise2} \cupow(A',B'),  on the product of 
\[ A' := AR\; , \; B' := R^\top B.\] 
The key point is that, while the matrices $RR^\top = I_n$ cancel at the \emph{output} $A'B' = ARR^\top B = AB$ (hence no need to ```peel-off'' the noise), 
this cancellation does \emph{not} occur locally on 
the \emph{partial sums}  
\begin{align}\label{eq_rand_partial_sums}
P_{I, K, J} = \sum_{i \in I, k \in K, j \in J} (AR)_{i k} (R^\top B)_{k j}
\end{align}
of  \emph{intermediate} $r\times r$ tiles ($|I|=|K|=|J|=r$).  

Of course, this scheme, described in Algorithm \ref{fig:RandRotationScheme}, is useless if the (dishonest) miner chooses the all-zero matrices $A=B=0^{n\times n}$.  Nevertheless, 
assuming the columns of $A$ and rows of $B$ have \emph{high rank} (which can be ensured by adding a random permutation of the identity matrix), say $\rk(A),\rk(B) \gtrsim n$ for simplicity), the FastJL transofm $R$ guarantees that every (say) $\sqrt{n} \times \sqrt{n}$ tile of $AR$ is a random \emph{subspace embedding} \cite{CW13}, so we may as well assume $A=B=I_n$.

\paragraph{Local-vs.-Global Structure of FFT} 
Assuming $A=B=I_n$ for simplicity,  
the problem of computing the \emph{partial sums}  $(RR^\top)_{I,J,K} \in \R^{r\times r}$ appears to be hard, since the \emph{global} structure of the randomized Hadamard matrix seems hard to exploit \emph{locally} on submatrices. 
Indeed, it is not known how to multiply arbitrary $r\times r$ \emph{submatrices} of the FFT matrix by general $x\in \R^r$ in $o(r^2)$ time, unless massive pre-processing is allowed \cite{KU11}. 

To ensure the \emph{marginal entropy} of every $r\times r$ tile of $(AR)\times (R^\top B)$ is full ($\sim r^2$ bits), Algorithm \ref{fig:RandRotationScheme} composes the Hadamard transform with a \emph{block}-diagonal matrix $D$ --   choosing the block-size to be $d=r$ would ensure this, but this will lead to overall runtime of $O(n^2(r+\log n))$ of Algorithm \ref{fig:RandRotationScheme}, no faster than the PoUW scheme from Algorithm \ref{fig:functionalmatmul} ($O(n^2r)$). Nevertheless, we conjecture that even block-size $d=O(1)$ remains computationally hard without massive preprocessing.   

\begin{conjecture}\label{conj:amortized_partial_sums}
Let $A,B \in \mathbb{R}^{n \times n}$ be full-rank matrices. Then  
the \emph{amortized} cost of computing all $n^3/r^3$ partial sums  \eqref{eq_rand_partial_sums} in Algorithm \ref{fig:RandRotationScheme} with $r= \sqrt{n}, d=O(1)$,  
even with $O(n^3)$ preprocessing time, 
requires  $\Omega\left(r^3\right)$ time, in the word-RAM model with word-size $w=O(\log n)$. 
\end{conjecture}

This is the analogue of Assumption~\ref{conj_DP_LR_scheme}.

\begin{myalg}[Random-Rotation PoUW Scheme (Self-Canceling Noise)]\label{fig:RandRotationScheme}
    \noindent \underline{$\mathsf{rrPoW}_{r,d}(\seed,A,B) $}: $\quad \sslash  \quad A,B\in \mathbb{R}^{n\times n}$
    \begin{enumerate}
        \item Parameters: $r \in \mathbb{N}$ (PoW tile-size); \;\; $d \in \mathbb{N}$ (Hadamard block-size).
        \item Parse 
        $\seed = D = \mathsf{Diag}(B_1, B_2, \dots, B_{n/d})$,
where each block 
$B_i \in \{-1,1\}^{d \times d}, \\
(B_i)_{jk} \sim \text{Unif}(\{-1,1\})$. 
$\quad \\\sslash   D \in \{-1,1\}^{n\times n}$ is a random block-diagonal matrix with block-size $d$.
\item Let $R \leftarrow H_n\cdot D$, where $H_n \in \{\pm 1\}^n$ is the $n\times n$ Hadamard matrix \cite{AC06}.
\item $A'\leftarrow AR \;\; , \;\; B'\leftarrow R^\top B$.
        \item Run $\cupow_r(A',B')$ (Algorithm \ref{fig:functionalmatmul}).
    \end{enumerate}
\end{myalg}

 \medskip 
 We leave additional theoretical and practical exploration of the (pseudo-)random rotation scheme for future work.



\section{Proof of Useful Work as a Poisson Process}
\label{sec:poisson}
Remember that a proof of useful work is associated with a function $f\colon \{0,1\}^n\to\{0,1\}^n$, reflecting the computation being done during the proof of work. 
We assume a canonical algorithm \textsc{Alg} for solving $f$, and denote $t(n,\tsk)$ the complexity of $\textsc{Alg}$ in solving $\tsk \in \{0,1\}^n$. We denote $t(n) = \max_{\tsk} f(n,\tsk)$ the worst-case complexity of $\textsc{Alg}$ over all inputs of size~$n$.

In a proof of work for blockchains, we require that the event that $\mathsf{Veify}$ outputs 1 follows a Poisson process. A Poisson process is a stochastic process that models the occurrence of random events over time or space, where these events occur independently and at a constant average rate. It is characterized by having exponentially distributed inter-arrival times between events. Bitcoin's proof of work mechanism is an example of a Poisson process wherein blocks are generated at a constant rate of~10 minutes.

\begin{definition}[Poisson process]
    A stream of events $E_1,E_2,\ldots$ is a Poisson process with rate $\rho$ if for every point in time $T$, interval $\overline T$, and parameter $k$, it holds that 
\[
\Pr\left[\sum_{i=T+1}^{T+\overline T} E_i = k\right] = \frac{(\rho \overline T)^k e^{-\rho \overline T}}{k!}.
\] 
\end{definition}

\begin{fact}\label{fact:bernoulli2poisson}
 Consider a sequence of $n$ independent and identically distributed Bernoulli random variables $E_1,\ldots,E_n$ with parameter $p$. (That is, each $E_i$ is 1 with probability $p$ and 0 with probability $0$.) The distribution of $\sum_{i=1}^n E_i$ follows a binomial distribution $B(n,p)$. It is known that the binomial distribution converges towards the Poisson distribution as $n\rightarrow \infty$ and $np$ converges to a finite limit (i.e., if $n$ is sufficiently large and $p$ is sufficiently small).
\end{fact}

\begin{definition}[Proof of useful work for blockchains]
\label{def:PoUWBlobkchain}
Let $\mathsf{Solve}(\seed, \mathsf{data}, \tsk)$ and $\mathsf{Verify}(\seed,\mathsf{data}, \pi)$ be two procedures, as in Section~\ref{sec:pouw-def}, except that we add the $\mathsf{data}$ parameter. These procedures constitute a $(\rho,\alpha)$  proof of work if the following two conditions hold:
\begin{itemize}
    \item[(a)] The honest mining process generates a stream of events $E_1,E_2,\ldots$ that forms a Poisson process with rate $\rho$.
    \item[(b)] For every $T$ and $\overline T$, the $(T,\overline T,\alpha)$-adversarial mining process generates a stream of events $E_1,E_2,\ldots$ that forms a Poisson process with rate $\le \rho$.
\end{itemize}

\setlength{\columnsep}{1cm} %
\setlength{\columnseprule}{0.4pt} 
\begin{mdframed}
\vspace{-10px}
\begin{multicols}{2}
\underline{\textbf{Honest mining process}}:
\begin{enumerate}
    \item Initialize $i=1$. 
    \item\label{item:pouw_bc}  Sample a uniformly random $\sigma_i$, and choose~$\tsk_i$ and~$\mathsf{data}_i$.
    \item Compute $\pi_i=\mathsf{Solve}(\seed_i,\mathsf{data}_i, \tsk_i)$.
    \item Compute $E_i = \mathsf{Verify}(\seed_i,\mathsf{data}_i,\pi_i)$.
    \item Increase $i$ by 1 and return to step~\ref{item:pouw_bc}.
\item[]
\item[]
\end{enumerate}

\underline {\textbf{$(T,\overline T,\alpha)$-adversarial mining process}}:
\begin{enumerate}
    \item Run the honest process for $T$ steps resulting with $E_1,\ldots,E_T$ and $\{(\sigma_i,\tsk_i,\mathsf{data}_i,\pi_i)\}_{i\in[ T]}$.
    \item 
    Given all of the above, initialize (in polynomial time) a time $\alpha(n)\cdot (t(n)\cdot \overline T)$ adversary. Upon a set of uniformly random $\sigma_{T+1},\ldots,\sigma_{T+\overline{T}}$, run it to generate $\pi_{T+1},\ldots,\pi_{T+\overline T}$. Compute $ E_i = \mathsf{Verify}(\seed_i,\mathsf{data}_i,\tilde \pi_i)$ for $T< i \le T+\overline T$.
    \item Continue running the honest process with $i=T+\overline T+1$.
\end{enumerate}
\end{multicols}
\end{mdframed}
\end{definition}

\paragraph{Bitcoin's proof of work.} 
To see that Bitcoin's PoW process forms a Poisson process with rate $\rho$, fix a point in time $t$ and a polynomially long interval $T$. A polynomial number of random oracle evaluations are independent of one another (and of any preprocessing). Thus, the number of successful random oracle evaluations (block discoveries) follows a Binomial distribution: $X \sim \text{Binomial}(n, p)$ with $n=R\cdot T$ and $p=\difficulty/2^\lambda$, where $R$ is the total hash rate (hashes per second) of the network and $\difficulty$ is a difficulty parameter. 
By Fact~\ref{fact:bernoulli2poisson}, for large $n$ and small $p$, the Binomial distribution approximates a Poisson distribution with parameter $\rho = n \cdot p = (R\cdot T \cdot \difficulty)/2^{\lambda}$. The value of  $\difficulty$ is adjusted so that $\rho=1/600$, implying a block every 10 minutes, in expectation. Since we are using a random oracle to model the hash function, it is also the case that iterated hashes is the best strategy for a malicious solver in the PoW, implying that there is no way to obtain any speedup.

\paragraph{Bitcoin from our PoUW.}
We argue that a PoUW for blockchains, as in Definition~\ref{def:PoUWBlobkchain}, along with standard properties of the cryptographic hash function, can be used to establish the security of the Bitcoin backbone protocol~\cite{GarayKL24}. More precisely, PoUW can be used to instantiate the same functionality as that of Bitcoin, i.e., a public ledger, as defined in~\cite{GarayKL24}. The proof follows the one presented in~\cite{GarayKP20}, where we instantiate signatures of work with PoUW.  We refer to~\cite{GarayKP20} for the proof.

\section{Unpredictable Encodings} \label{appx:ue}
In Section~\ref{sec:tue} we introduced a notion termed transcript unpredictable encodings that suffices for a proof of useful work for matrix multiplication. Here we provide an even more generic abstraction that we believe could be of independent interest. We call this primitive an \emph{unpredictable encoding} scheme (UES) which is a form of an ``instance optimal'' randomized encoding scheme with some new and incomparable properties to classical randomized encoding schemes. 
 A UES is associated with a function $f\colon \{0,1\}^n\to\{0,1\}^*$ with a canonical algorithm \textsc{Alg} with worst-case running time $t(n)$, and it consists of two operations:
\begin{itemize}
    \item $\mathsf{Encode}(x;r )$: for an input $x$ this probabilistic procedure  outputs an encoding $\tilde x$.
    \item $\mathsf{Eval}(\tilde x)$: for an encoded input $\tilde x$ this deterministic procedure outputs an output $y$.
\end{itemize}

\paragraph{Correctness.} We require that $\mathsf{Eval} (\mathsf{Encode}(x)) = f(x)$ for every $x$.

\paragraph{Running time.}
 The worst-case running time of $\mathsf{Encode}$ plus the running time of $\mathsf{Eval}$ is $t(n)\cdot (1+o(1))$.

\paragraph{Hardness.} It should be hard to compute the encoding of a given $x$ on a random $r$ in $\ll $ the running time of $\mathsf{Encode}$. That is, there is a constant $C$ such that for every $x$ and every time $t'(n)$ algorithm $\mathcal A$, it holds that $
    \Pr_r[ \mathcal A(x,r) = \mathsf{Encode}(x;r)] \le C\cdot (t'(n)/t(n))$.

Notice that the transcript unpredictability definition from Section~\ref{sec:tue} is a special case of the above where the encoding $\tilde x$ is essentially the transcript of the canonical matrix multiplication algorithm applied on the encoded input matrices. 

\paragraph{Relation to randomized encodings.}
The above definition resembles, at least syntactically, randomized encodings~\cite{DBLP:conf/focs/IshaiK00,DBLP:journals/siamcomp/ApplebaumIK06}, and so we explain the differences. A randomized encoding allows to express an arbitrary computation, given by a function $f$ and
input $x$, by a randomized representation $\hat{f}(x)$ whose distribution (viewed as a random variable) encodes
$f(x)$, while revealing nothing else regarding $f$ and~$x$. We do not require privacy of $x$ given $\tilde x$. On the other hand, we require a fine-grained hardness property and also an ``optimal'' efficiency property.

\end{document}